\documentclass[runningheads, 10pt]{llncs}
\usepackage{graphicx}
\usepackage{subfigure}
\usepackage{bibnames}
\usepackage{amsmath,amssymb,amsfonts}
\usepackage[ruled]{algorithm2e} 
\usepackage{multirow,array}
\usepackage{floatrow}
\newfloatcommand{capbtabbox}{table}[][\FBwidth]
\usepackage[utf8x]{inputenc} 
\DeclareUnicodeCharacter{65292}{okay}

\include{mydef}

\newtheorem{thm}{Theorem}
\newtheorem{mydefinition}[thm]{Definition}

\begin{document}
\title{Bribery in Rating Systems: \\A Game-Theoretic Perspective}

\author{Xin Zhou \inst{1} \and 
	Shigeo Matsubara  \inst{2} \and 
	Yuan Liu\thanks{Corresponding authors.} \inst{3} \and
	Qidong Liu$^\star$ \inst{4}
}
\institute{School of Computer Science and Engineering, Nanyang Technological University, Singapore.
\email{xin.zhou@ntu.edu.sg} \and 
Center for Mathematical Modeling and Data Science, Osaka University, 
Osaka, Japan. \email{matsubara@sigmath.es.osaka-u.ac.jp} \and  
Cyberspace Institute of Advanced Technology, Guangzhou University，\\ Guang Dong, China.
\email{liuyuan@swc.neu.edu.cn} \and
School of Computer and Artificial Intelligence, Zhengzhou University, \\ Zhengzhou, China. \email{ieqdliu@zzu.edu.cn}
}

\authorrunning{X. Zhou et al.}

\maketitle              
\begin{abstract}
Rating systems play a vital role in the exponential growth of service-oriented markets.
As highly rated online services usually receive substantial revenue in the markets, malicious sellers seek to boost their service evaluation by manipulating the rating system with fake ratings.  
One effective way to improve the service evaluation is to hire fake rating providers by bribery.
The fake ratings given by the bribed buyers influence the evaluation of the service, which further impacts the decision-making of potential buyers.
In this paper, we study the bribery of a rating system with multiple sellers and buyers via a game-theoretic perspective.
In detail, we examine whether there exists an equilibrium state in the market in which the rating system is expected to be bribery-proof: no bribery strategy yields a strictly positive gain.
We first collect real-world data for modeling the bribery problem in rating systems.
On top of that, we analyze the problem of bribery in a rating system as a static game.
From our analysis, we conclude that at least a Nash equilibrium can be reached in the bribery game of rating systems.

\end{abstract}

\section{Introduction}\label{sec:introduction}

In e-marketplaces ($e.g.$, Amazon, app store), buyers usually select reputable services or items offered by sellers based on their direct experiences or ratings given from other buyers. 
Direct experience is important and trustworthy; however, it is infeasible for buyers to interact with all sellers in order to gain direct experiences in a real-world setting. 
As a result, a rating/reputation system evaluating the given ratings is proposed as a crucial tool for supporting buyers’ decision-making~\cite{josang2007survey}.
In our context, we assume a seller and the owned service are one-to-one mapping.
Hence, when a buyer interacts with one seller, he/she purchases the item or uses the service from that seller.

In real-world rating systems, various studies have demonstrated that high profit is attached to high reputation services or items. A study of eBay conducted by Resnick $et~al$. revealed that buyers were willing to pay 8\% more to the item offered by reputable sellers than that provided by new sellers~\cite{resnick2006value}. Ye $et~al.$~\cite{ye2009impact} in their study showed that a 10\% improvement in reviewers’ rating on hotels could increase the number of bookings by 4.4\%. Driven by profit, sellers have incentives to perform unfair rating attacks by bribing buyers and increase their aggregated rating above their competitor’s, as this would increase their market share and overall revenues thereof~\cite{zhu2014discovery,liu2020fraudt,grandi2020personalised}. Fake ratings have been reported on the real-world movie dataset collected by Cao $et~al.$~\cite{cao2016complete}, and the fake ratings dominate the rating distribution in the first few days as soon as the movie is released. The fake ratings are mainly divided into fake positive ratings to boost the sellers’ reputation and dishonest negative ratings to decrease their competitors’ reputation. Although service attack detection frameworks or models have been proposed recently~\cite{manzoor2017attackdive,ouffoue2017attack,liu20213r,saude2021robust}, emerging underground market (in providing fake ratings) even worsen the bribery problem~\cite{sampath2018corporate}.

We assume there are two potential reasons that entice the sellers into bribing ``bot farms'' or ``human water armies''~\cite{wang2012serf} in order to inflate service ratings. The first is the low cost of performing
the bribery due to rating sparsity~\cite{zhou2016evaluating}.
From the analysis on the Google Play app store, the rating ratio is as low as 2.1\% based on a study of the top 20,000 apps~\cite{zhou2019arm}. That is, about 2 out of 100 buyers or users of the app rate their ratings on the platform.
In such a situation, app owners can easily increase the positive evaluation of their apps by bribing buyers with a compensation, monetary or not. That is, sellers exchange for positive ratings by bribing a group of buyers with certain payment. The positive ratings further influence the sellers' future revenues. 
The second reason is the attractive revenue gained from temporal (usually on a daily basis) top app ranking. App developers resort to fraudulent means to deliberately boost up their app ranking on an App store~\cite{zhu2014discovery,ramos2020negative}.
In this paper, we analyze the effect of a bribing strategy by defining the utility function and strategy space for each seller.  In addition, we model the problem as an economic game and examine whether an equilibrium state exists, in which the rating system is expected to be bribery-proof: no individual seller wants to perform any bribery strategy.

In single-player decision theory, the key notion is that of an optimal strategy, that is, a strategy that maximizes the player’s expected payoff for a given environment in which the player operates~\cite{lianju2011game,Grandi2016}. Such a situation in the single-player case can be fraught with uncertainty, since the environment might be stochastic, partially observable, and spring all kinds of surprises on the player.
In our work, we extend the setting into a more complex multi-players setting, in which the environment comprises other players, all of whom are also willing to maximize their payoffs. Thus, the notion of an optimal strategy for a given player is not meaningful as the best strategy also depends on the choices of others. The equilibrium state is reached when all players perform their best strategy considering others players' strategies.
In the multi-player setting, we model the sellers as players and map the bribing behavior into strategy space in game theory. Based on the model, we intend to find the equilibrium state within all combinations of sellers' strategies.
To simplify our analysis, we will use the terms ``seller'', ``player'' or ``service'' interchangeably.


In detail, we consider the static game in which the number of sellers and buyers are fixed within a time slot. The exponential increase of the strategy space with the number of sellers and buyers makes it difficult to search the optimal strategies. To find the equilibrium state in the static game, we devise a greedy algorithm to eliminate the dominated strategies.

We summarize the main contributions as follows:
\begin{itemize}
	\item
	Although some researchers have studied bribery in the rating system, they modeled the problem without real-world data support and merely depended on assumptions.
	In this paper, we collect two real-world datasets as the backbone of our modeling and analysis.
	\item We extend bribery in rating systems from a single seller into multiple sellers and buyers situation and explicitly interpret the problem from a game theory perspective.
	\item On top of the model, we analyze the effect of bribing strategies in the case of static game and found that a Nash equilibrium exists in sellers' strategy space, with too many buyers are bribed under Nash equilibrium compared to the social optimum.
\end{itemize}

We reveal that our study can be applied to all rating-based systems in which individuals may influence the decision of one other via the rating value. Bribery can disrupt buyers to interact with fake reputable sellers with a bribed high rating value.

%
%


\textbf{Related research lines}:
Our approach relates to several research lines in computational social choice, game theory, and reputation systems. 

\begin{description}
	\item[Computational social choice:] In our studied rating system, sellers try to influence potential buyers' decisions to maximize their own expected payoff. The problem is analogous to lobbying and bribery in computational social choice, where agents manipulate or modify their outcomes to reach their objectives~\cite{faliszewski2009hard,faliszewski2009llull}. Lobbying and bribery are established concepts in computational social choice, the research ranges from the seminal contribution of~\cite{helpman2001lobbying} to recent studies such as vote-buying in~\cite{Parkes2017Thwarting}.
	\item[Reputation systems:] Dishonest ratings influence the accuracy of reputation evaluation; hence, a robust reputation system should detect the attacks and mitigate the influence caused by the attacks~\cite{xu2017online}. Otherwise, the disruptive reputation can influence other agents. In this sense, ours can be seen as a study of reputation in multi-agent systems~\cite{jiang2013evolutionary}.
	To analyze the harm of unfair ratings, more recent studies based on information theory have been proposed in papers~\cite{wang2015quantifying} and~\cite{wang2016harmful}.
\end{description}


%
%

\begin{figure*}
	\centering
	\subfigure[\# of reviews \textit{vs.} ratings on Apple Store.\label{sfig:snowball_a}]{\includegraphics[trim=5 0 10 0, clip, width=0.31\textwidth]{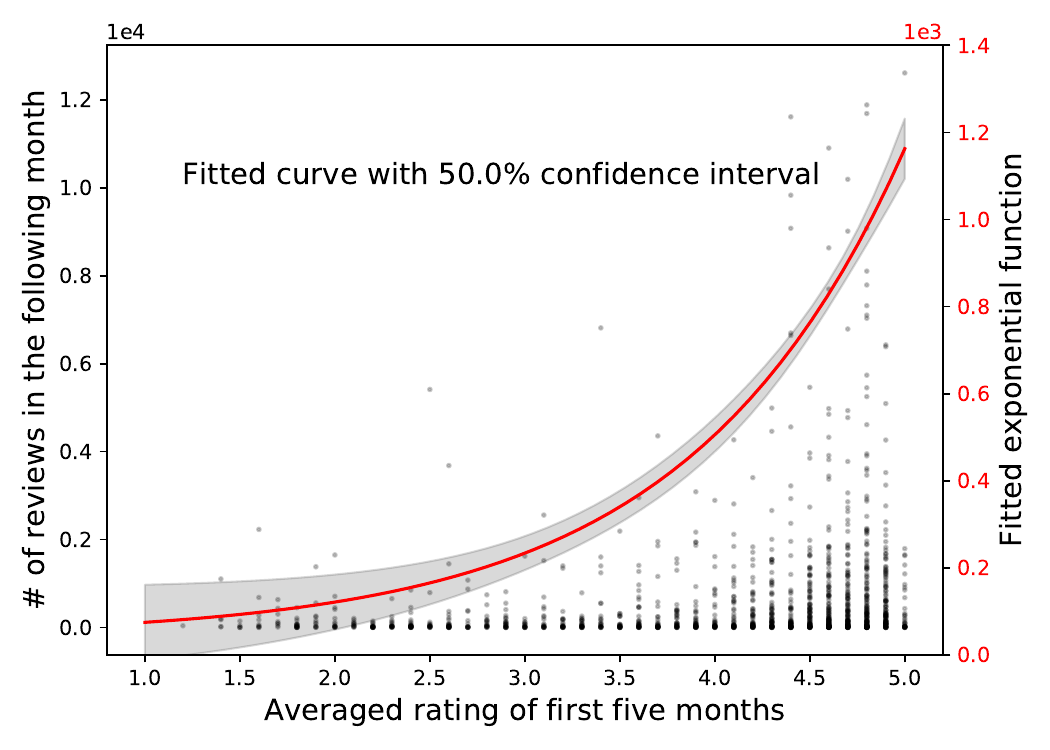}}
	\hfill
	\subfigure[Different genres of apps from Google Play. \label{sfig:snowball_b}]{\includegraphics[trim=5 0 10 10, clip, width=0.31\textwidth]{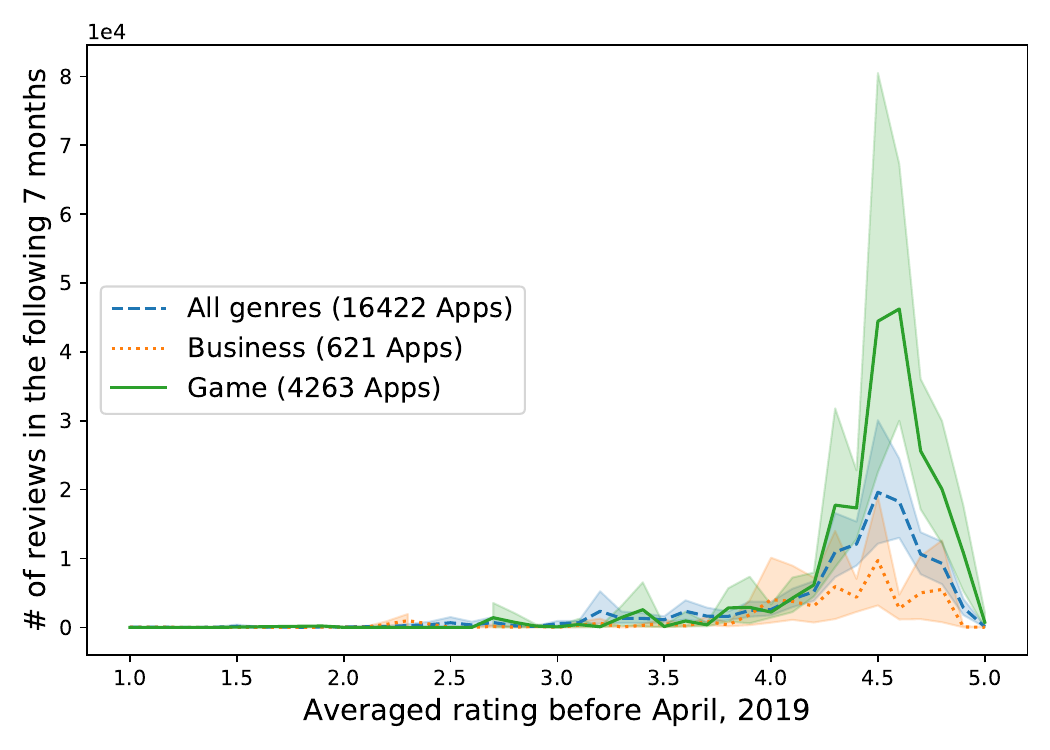}}
	\hfill
	\subfigure[\# of reviews \textit{vs.} installs on Google Play. \label{sfig:snowball_d}]{\includegraphics[trim=5 0 10 0, clip, width=0.31\textwidth]{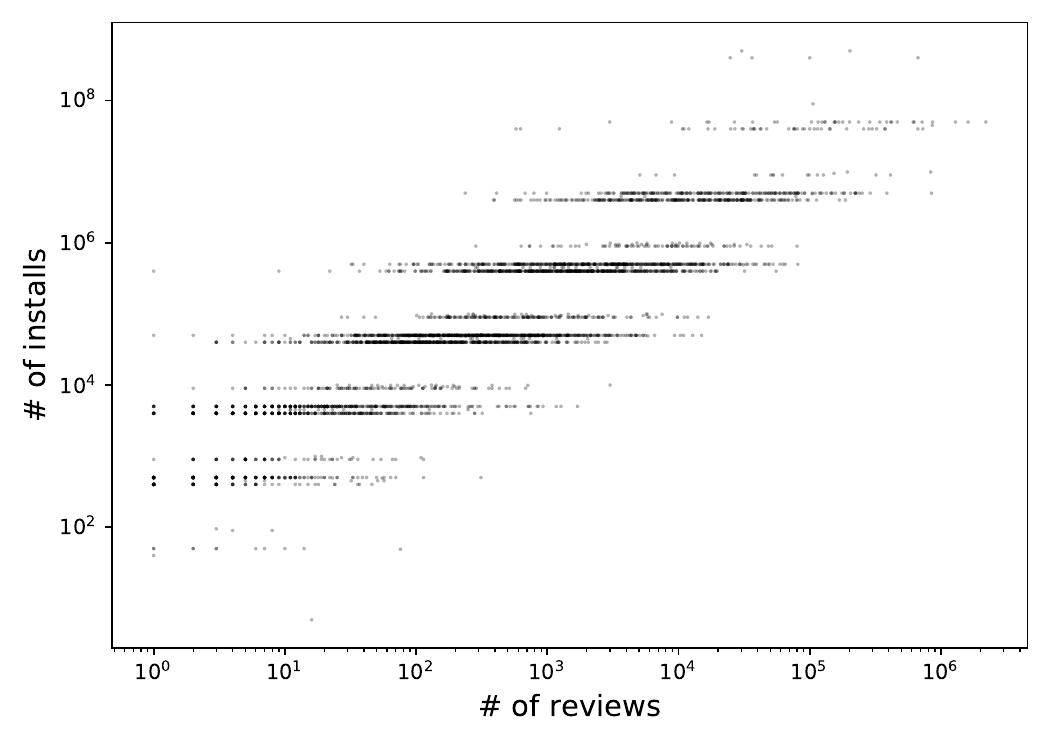}}
	\caption{Snowball effect shown on apps in different stores and different genres.}
	\label{fig:snowball}
\end{figure*}

\section{Profit Modeling on Data Observation}
We collect two real-world datasets to show how the problem can be represented as an economic game. 
By crawling the statistical data of 1609 apps from Apple App Store during December 2016 to May 2017 and the data of rating, app genre, reviews, and installs on 16422 apps from Google Play during April 2019 to November 2019, we first model the relation between the received rating of an app and its potential profit in the number of reviews.
Although individual apps may exhibit different patterns between ratings and installs, we study whether the overall pattern is consistent between apps from different app stores or app genres.
From the plotted Fig.~\ref{fig:snowball} (a) and (b), we find that higher-rated apps receive more reviews, although they are in different app stores and within different app genres.
The data shows that the relation between the received rating $r$ and the number of following reviews $N_{re}$ follows a power function.
\begin{equation}
N_{re} = a \cdot r^n,
\label{eq:revr}
\end{equation}
where $a \neq 0$ and $n$ are real numbers, and $a$ is a coefficient.
Since not all buyers score their rating after purchase, our collected data may not accurately capture the relation between rating and the transaction volume.
Seller's profile may also be evaluated in the number of installs; we further show the number of installs on a given rating in Fig.~\ref{fig:snowball} (c).
From Fig.~\ref{fig:snowball} (c), we can observe that the relation between the received rating $r$ and the number of following installs $N_{in}$ also follows a power function.
However, compared with the number of reviews, the number of installs is much higher skewed with ratings.
In Google Play, the number of installs is merely an estimated value on the order of magnitude of downloads.
Actually, we can obtain the relation between rating and installs through Fig.~\ref{fig:snowball} (b) and (c) indirectly.
Fig.~\ref{fig:snowball} (c) shows that the log scale of the number of installs $N_{in}$ is proportional to the log scale of the number of reviews $N_{re}$.
That is,

%
\begin{equation}
log N_{in} = b \cdot log N_{re} \Leftrightarrow N_{in} = (N_{re})^b,
\label{eq:revinstall}
\end{equation}
where $a \neq 0$.
Substituting equation~\eqref{eq:revr} into equation~\eqref{eq:revinstall}, we have
\begin{equation}
N_{in}=k \cdot r^{\Omega},
\label{eq:snowballeff}
\end{equation}
where $k=a^b$ and $\Omega = n\cdot b$.
The number of installs based on ratings also shows a snowball effect; when a few buyers explore a good app, everyone else wants to pile it on.
Following the observation, app owners competed to increase their app evaluation in order to increase their purchase revenue.
For a single seller, with limited potential buyers in the platform, maximizing its own payoff will influence the other's  welfare.
Hence, the problem can be represented as an economic game.




\section{Static Game}
To present our model, we assume that, in the rating system, each seller can observe the current buyers of other sellers alongside the current evaluation from the buyers.
The common knowledge in our game is as follows: \textit{a)} The players are rational and act with the intention to maximize their revenue. \textit{b)} In static game, all the players know their own payoffs as well as others, also all the players know that all the players know all the players' payoffs, and all the players know that all the players know that all the players know all the players' payoffs, and so on, ad infinitum.
We let the number of buyers be constant in a static game, the bribery space for a seller is the buyers that have not interacted with that seller.
The following analysis takes the rating system in the app store as an example; the analysis also can be applied to other numerical rating systems.

Considering that we have $M$ \textit{sellers} and $N$ \textit{buyers} in an app store within a time slot.
After a seller $s_i$ released its app on the app store, the app is evaluated by a finite set of buyers
$B_i=\{b_1, b_2, ..., b_l\}, l <= N$.
Low-rated app owners intend to bribe a set of buyers to increase its app evaluation.
We assume one buyer can only be bribed by one seller and the apps in the market are functionally equivalent.
The bribed rating is proportional to the payment, that is, more payment for higher rating.
For a fair rater, based on his/her experience, he/she submits a rating to evaluate the app.
The scored rating $r$ is drawn from a set of values $r \in [0, 1]$.
Although discrete assignment of 1 to 5 stars is common in online rating systems, such as Amazon, Apple Store, and Google play, the stars can be mapped into $[0, 1]$ to simplify the analysis.

We study the rating systems such as used in Apple Store; every interested buyer can see the evaluation of other buyers.
Only the buyer who paid for the app can evaluate it, and the aggregation of the evaluation will influence potential buyers' decision.
We represent the aggregated evaluation on the app of seller $s_i$ as a function $\bar{r_i}=avg(R_i)$ of the received ratings $R_i=\{r_1, r_2, ..., r_l\}$, where $avg$ is the average function across all real-valued ratings.

Suppose in a time slot, all the sellers perform bribery strategy to increase their app evaluation; as a result, the profit is based on their current rating distribution.
We define the strategy of sellers in the form of an ensemble of individual strategies:
\begin{mydefinition}
	\label{def:stra}
	An individual strategy on buyer $b_j$ is the amount of effort $\phi(b_j)$ required to improve its current rating to a new rating.
	We denote: $\phi(b_j) \rightarrow [0, +\infty)$. A strategy of a seller is an ensemble of individual strategies on each buyer, such that $\sum_{j=0}^{N} b_j \in [0, +\infty)$.
\end{mydefinition}

Under bribery, buyer's evaluation is closed under $min\{1, x+y\}$ for all $x, y \in [0, 1]$, where $x, y$ are two evaluations on an app.
To compare with the results of ~\cite{Grandi2016}, we adopt that the bribed cost is proportional to buyer's rating. For example, buyer $b_j$'s evaluation is 0.5, but when bribed by seller $s_i$ with $\phi(b_j) = 0.6$, the new evaluation will be $min\{1, x+y\} = min\{1, 1.1\} = 1$.
By definition~\ref{def:stra}, the total cost for seller $s_i$ to perform a bribery on a set of bribed buyers $B_{\phi}$ $(|B_{\phi}| \le N)$ is $\sum_{b\in B_{\phi}}\phi(b)$.
Next, we define the potential utility to a seller $s_i$ that can be gained from bribing a set of buyers.
To simplify our analysis, we assume the utility a seller can obtain is proportional to its aggregated rating.
That is, we let $\Omega = 1$ in Equation~\eqref{eq:snowballeff}. The assumption does not influence our final conclusion.

\begin{mydefinition}
	\label{def:utility}
	The expected utility of a seller $s_i$ is a function.
	\begin{equation}
	u(\bar{r_i})=\mathcal{N} \cdot \bar{r_i}^{\Omega} \cdot k,
	\end{equation}
	where $\mathcal{N} = N - \sum_{j=0}^M |B_j|$ is the potential buyers in the market,  $k$ can be viewed as the profit from an app purchase.
\end{mydefinition}
The utility of seller $s_i$ amounts to the number of potential buyers in the market and the rating rated by the interacted buyers. The underlying intuition is that buyers prefer to download and use higher rated apps.
Theoretically, a seller can bribe all the available buyers in the market; however, the potential profit will approach to zero because of no future interaction.
Based on previous definitions, we can define the payoff of a bribery strategy.

\begin{mydefinition}
	\label{def:payoff}
	Let $\phi$ be a strategy and $B_\phi$ is the bribed buyers by $\phi$. The payoff to seller $s_i$ from $\phi$ is:
	\begin{displaymath}
	\pi(\phi) = u^{\phi} - u^{0} - \sum_{b \in B_\phi} \phi(b),
	\end{displaymath}
	where $u^{0}$ is the initial utility and $u^{\phi}$ is the utility after execution of $\phi$.
\end{mydefinition}

\begin{mydefinition}
	\label{def:domi}
	Let $\phi'$ and $\phi''$ be feasible strategies for seller $s_i$. Strategy $\phi'$ is strictly dominated by $\phi''$ if for each feasible combination of the other players' strategies, $s_i$'s payoff from playing $\phi'$ is less than $s_i$'s payoff from playing $\phi''$:
	\begin{align*}
		\pi(\phi_1, ..., \phi_{j-1}, \phi', \phi_{j+1}, ..., \phi_{\mathcal{N}}) < 
		\pi(\phi_1, ..., \phi_{j-1}, \phi'', \phi_{j+1}, ..., \phi_{\mathcal{N}}),
	\end{align*}
	for each $(\phi_1, ..., \phi_{j-1}, \phi_{j+1}, ..., \phi_{\mathcal{N}})$ that can be constructed from the other players' strategy spaces $\varPhi_1, ..., \varPhi_{j-1}, \varPhi_{j+1}, ..., \varPhi_{\mathcal{N}}$.
	Let $\phi_{-i}$ denote the strategy combination from players other than $j$. Then, if occasionally, $\pi(\phi', \phi_{-j})$ = $\pi(\phi'', \phi_{-j})$, we denote $\phi'$ is weakly dominated by $\phi''$.
\end{mydefinition}

\begin{algorithm}[!h]
	\caption{The greedy bribing strategy}
	\label{alg:greedy}
	\SetKwProg{greedy}{Function \emph{greedy}}{}{end}
	
	\textbf{Input}: Seller $i$'s evaluation set $B_i, R_i, \bar{r_i}$, budget $\mathcal{B}$ \;
	\textbf{Output}: A bribing strategy $\phi$\;
	\greedy{($B_i, R_i, \bar{r_i}, \mathcal{B}$)}{
		Sort $r_j \in R_i$ in ascending order $R_D = r_0, ..., r_l$\;
		\ForEach{$r_j~in~R_D$}{
			\If {$\mathcal{B} > 0$}{
				\If {$r_j < 1$}{
					$\phi(j) = min\{1-r_j, \mathcal{B}\}$\;
					$\mathcal{B} = \mathcal{B} - \phi(j)$
				}
			}
		}
		\If{$\mathcal{B} > 0$} {
			Compute non-rater set $B_{-i}$\;
			\ForEach{$j~in~B_{-i}$}{
				\If {$\mathcal{B} > \bar{r_i}$}{
					$\phi(j) = min\{1, \mathcal{B}\}$\;
					$\mathcal{B} = \mathcal{B} - \phi(j)$
				} else {
					break
				}
			}
		}
	}
\end{algorithm}

Given the definitions and the conclusions in paper~\cite{Grandi2016}, a rational seller would prefer greedy strategy to others under the same budget.
In other words, a strategy of bribing non-raters always dominates others.
Algorithm~\ref{alg:greedy} defines the greedy bribing strategy.
In the first loop, the seller bribes the history buyers who scored low ratings greedily.
If the budget is still above zero, then the seller tries to bribe new buyers.


\begin{proposition}
	\label{prop:greedy}
	Under same combination strategies $\phi_{-i}$ of other sellers', the greedy bribing strategy $\phi^o$ for $s_i$ always dominates others with same budget $\mathcal{B}$.
\end{proposition}

	
	

Under the same budget, it shows when a buyer is bribed, it is worthy to bribe it completely with the highest rating.
We further analyze under what conditions the rating system is bribery-proof.
Based on the above proposition, we can show that no strategy is profitable if the number of potential buyers is too low, $i.e.$, the ratio between the number of current rated buyers $|B_i|$ and potential buyers is greater than the profit of an item $k$.
\begin{lemma}
	\label{lemma:profitstra}
	For seller $s_i$, no strategy is profitable if
	\begin{displaymath}
	\frac{|B_i|}{\mathcal{N}} \geq k.
	\end{displaymath}
\end{lemma}

\begin{proof}
	Let $\phi$ be any strategy for seller $s_i$ and $s_i$'s current rated buyers $B_i \neq \emptyset$ . Then, by definition~\ref{def:payoff}, the payoff under $\phi$ with bribing budget $\mathcal{B} > 0$ is
	\begin{displaymath}
	\pi(\phi) = u^{\phi} - u^{0} - \mathcal{B}.
	\end{displaymath}
	According to proposition~\ref{prop:greedy}, the greedy strategy always dominates others.
	We give proof for the greedy strategy here.
	\begin{align*}
		& \pi(\phi) =
		u^{\phi} - u^{0} - \mathcal{B} = \\
		& \mathcal{N}\frac{\mathcal{B} + \sum_{j \in R_i} r_j}{|B_i|}\cdot k - \mathcal{N}\frac{\sum_{j \in R_i} r_j}{|B_i|}\cdot k - \mathcal{B} =
		\mathcal{N}\cdot k \cdot (\frac{\mathcal{B}}{|B_i|}) - \mathcal{B}.
	\end{align*}
	
	If strategy $\phi$ is profitable, then
	$\mathcal{N} \cdot k \cdot (\frac{\mathcal{B}}{|B_i|}) - \mathcal{B} > 0.$
	Therefore, for a profitable strategy, we have
	$\frac{|B_i|}{\mathcal{N}} < k.$
\end{proof}

The lemma shows for sellers that have already bribed a large amount of buyers, further bribing may result in no profit.
In real world, a bad reputed seller is more willing to improve its evaluation by performing a bribery strategy.
In multi-player situations, an open question for a seller is whether the best response to other player's strategies exists.
In the next section, we analyze the strategy space for each player and identify the best response in multi-player situations.

\subsection{Best Response}

In the previous section, we showed that bribing more buyers may not result in more profit.
When a seller bribes a buyer, it is better to induce the buyer to give the highest rating because the greedy strategy always dominates other strategies.
However, in real world, a buyer who has given his/her rating may refuse to improve it.
Thus, in this section, we redefine the game in this situation.

Based on Lemma~\ref{lemma:profitstra}, there exists a maximum number of buyers that can be bribed by seller $s_i$ in the market.
For seller $s_i$, it is always profitable to bribe the interacted buyers as shown in Lemma~\ref{lemma:profitstra}.
For low $\bar{r_i}$, sellers will try to bribe new buyers to get the highest rating, 1.
Thus, we redefine the strategy of a seller: a strategy $\phi$ for seller $s_i$ is the number of buyers that will be bribed $\phi_i$.
The strategy space $[0, \mathcal{N}]$ covers all the choices that could be of interest to $s_i$.
The redefinition of the strategy implies that the optimal strategy always
bribes buyers with the highest rating.
The new definition focuses on buyers instead of ratings.

Given the seller $s_i$'s current state, aggregated ratings $\bar{r_i}$, and the number of potential buyers in the market,
some interacted buyers refused to improve their given rating; thus, we are curious about how many buyers should be bribed for maximum profit.


The payoff to seller $s_i$ from bribing $\phi_i$ buyers when the bribery strategy combination of the other sellers are $(\phi_1, ..., \phi_{i-1}, \phi_{i+1}, ..., \phi_M)$ is
\begin{equation}
\label{equation:strax}
\pi(\phi_i, \phi_{-i}) = u^{(\phi_i, \phi_{-i})} - u^{0} - \phi_i,
\end{equation}
where $u^{(\phi_i, \phi_{-i})}$ is the utility for seller $s_i$ after $s_i$'s execution of $\phi_i$ and the others' execution of $\phi_{-i}$, respectively.

\begin{thm}
	\label{theorem:nash}
	The bribery game $G$ in a rating system has at least one Nash equilibrium.
\end{thm}
\begin{proof}
	From Nash theorem~\cite{nash1950equilibrium}, every game with a finite number of players and action profiles has at least one Nash equilibrium. In a bribery game, both the numbers of sellers and the strategy space are finite. Therefore, at least one Nash equilibrium can be reached in the game.
	
	Let $(\phi^*_1, ..., \phi^*_M)$ be a Nash equilibrium in $G$. Then, for each $\phi^*_i$, $\phi^*_i$ maximizes equation~\eqref{equation:strax} given that the other sellers choose strategy combination $(\phi^*_1, ..., \phi^*_{i-1}, \phi^*_{i+1}, ..., \phi^*_M)$.
	The first-order condition for this optimization problem is
	\begin{equation}
		\label{equation:stra}
		\frac{\partial \pi(\phi_i, \phi_{-i})}{\partial \phi_i} =  \frac{\partial u^{(\phi_i, \phi_{-i}})}{\partial \phi_i}  - 1 = 0  \Rightarrow \frac{\partial u^{(\phi_i, \phi_{-i})}}{\partial \phi_i} = 1.
	\end{equation}
	
	Based on Lemma~\ref{lemma:profitstra}, for seller $s_i$, if $\frac{|B_i|}{\mathcal{N}} \geq k$, the optimal strategy for $s_i$ would be $\phi_i = 0$.
	The best response $(\phi^*_1, ..., \phi^*_M)$ will solve
	
	\begin{align}
		\label{equation:matrix}
		\begin{split}
			\frac{\partial u^{(\phi_1, \phi_{-1})}}{\partial \phi_1} = 1;  \cdots = 1; 
			\frac{\partial u^{(\phi_M, \phi_{-M})}}{\partial \phi_M} = 1.
			\end{split}
	\end{align}
	Subject to
	\begin{align}
		\label{equation:cons}
		\begin{split}
			\phi_i \in \mathbb{N};
			\sum \phi_i <= \mathcal{N} \mathrm{~and~}
			\phi_i = 0;~if~\frac{|B_i|}{\mathcal{N}} \geq k.
		\end{split}
	\end{align}
\end{proof}

\textbf{Example 1.} Consider two sellers of a duopoly game in a market, in which the number of interacted buyers with seller $i$ is $B_i=5$ and seller $j$ is $B_j=2$.
These buyers cannot be bribed again by other bribery strategies.
The average evaluation given by those buyers on sellers $i$ and $j$ are $\bar{r_i}=0.2$ and $\bar{r_j}=0.5$, respectively.
The total number of buyers in the market is $n = 20$.
The profit for bribing one buyer is $k = 2$.
The strategies for the sellers are $\phi_i$ and $\phi_j$.
If $(\phi_i, \phi_j)$ forms a Nash equilibrium, according to the above statements, the optimal strategy should solve equation~\eqref{equation:stra}.

Substituting $\phi_i$ in equation~\eqref{equation:stra} and expanding $u^{\phi_i}$ yields

\begin{equation}
k \cdot [(\mathcal{N} - \sum \phi_i)\frac{\bar{r_i}\cdot |B_i| + \phi_i}{|B_i| + \phi_i}]' = 1.
\end{equation}
Each $\phi_i$ can be solved by computing all the sellers' first-order conditions subjecting to the constraints of equation
~\eqref{equation:cons}.
%
Solving the partial differential equations considering the constraints yields
$\phi_i = 2, \phi_j = 1$.

In this example, the first four strategies for both sellers dominant others; we further present the  bribery problem in the accompanying bi-matrix Table~\ref{table:matrix}.
The above result can also be derived by iterated elimination of strictly dominated strategies. In Table~\ref{table:matrix}, both seller $i$ and $j$ have four strategies. For seller $j$, strategy 1 strictly dominates the others, so a rational seller $j$ will play strategy 1 anyway. Thus, if seller $i$ knows that seller $j$ is rational, then seller $i$ can eliminate strategies 0, 2, and 3 from seller $i$ strategy space. In the second column, seller $i$'s strategy 2 strictly dominates the others, leaving (6.57, 12.33) as the outcome of the game.
We also plot seller $i$'s payoff in Fig.~\ref{fig:payoff3d}. In the figure, when the strategy is larger than 4, the payoff is decreased linearly. This is consistent with Lemma~\ref{lemma:profitstra} that bribing more buyers may not yield more profit.


\begin{figure}
\begin{floatrow}
\ffigbox{%
  \includegraphics[trim=90 25 35 60, clip, width=0.4\textwidth]{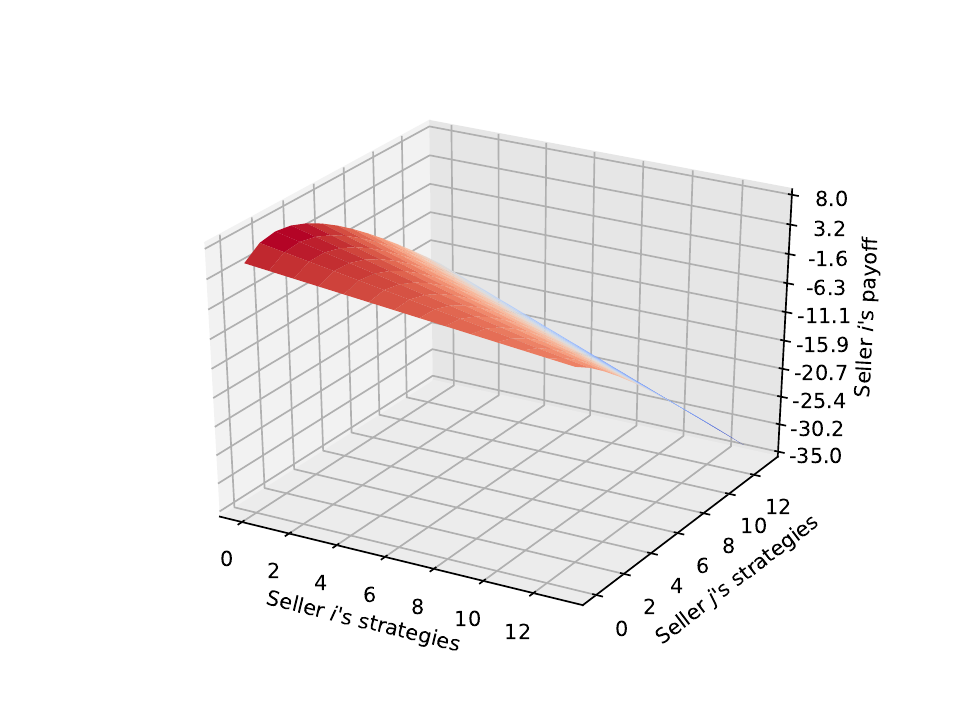}
}{%
  \caption{Seller $i$'s payoff.}%
  \label{fig:payoff3d}
}
\capbtabbox{%
\begin{tabular}{l c|c|c|c|c|}
		& \multicolumn{1}{c}{} & \multicolumn{4}{c}{$j$}\\
		& \multicolumn{1}{c}{} & \multicolumn{1}{c}{$0$}  & \multicolumn{1}{c}{$1$} & \multicolumn{1}{c}{$2$} & \multicolumn{1}{c}{$3$} \\\cline{3-6}
		\multirow{4}*{$i$}  & $0$ & $5.20, 13.00$ & $4.80, 15.00$ & $4.40, 14.50$ & $4.00, 13.00$  \\\cline{3-6}
		& $1$ & $7.00, 12.00$ & $6.33, 13.67$ & $5.67, 13.00$ & $5.00, 11.40$ \\\cline{3-6}
		& $2$ & $7.43, 11.00$ & $6.57, 12.33$ & $5.71, 11.50$ & $4.86, 9.80$ \\\cline{3-6}
		& $3$ & $7.00, 10.00$ & $6.00, 11.00$ & $5.00, 10.00$ & $4.00, 8.20 $ \\\cline{3-6}
	\end{tabular}
}{%
  \caption{A bribery model of duopoly}%
  \label{table:matrix}
}
\end{floatrow}
\end{figure}

In contrast to individual maximized payoff, the social optimum considers the whole social welfare and maximizes
\begin{equation}
\label{equation:sopt}
\max\limits_{0\leq O^{**} < \infty}O^{**} =  \max\limits_{0\leq \phi_i^{**} \leq N} \sum_{i = 1}^M\pi(\phi_i^{**}, \phi_{-i}),
\end{equation}
where $O=\sum_{i = 1}^M\pi(\phi_i, \phi_{-i})$.
Based on Theorem~\ref{theorem:nash}, we can infer this corollary.
\begin{corollary}
	Too many buyers are bribed in the Nash equilibrium compared to the social optimum. That is, $\mathcal{O^*} \geq \mathcal{O^{**}}$.
\end{corollary}
\begin{proof}
	We prove the corollary by contradiction and assume $\mathcal{O^*} < \mathcal{O^{**}}$.
	Then, there must exist at least one seller indexed by $i' (1 \leq i' \leq M)$, that $\phi_{i'}^{**} > \phi_{i'}^*$, since $\sum_{i=1}^M \phi_i^{**} = O^{**} > O^* = \sum_{i=1}^M \phi_i^*$.
	Because  $\phi_{i'}^*$ is the number of bribed buyers by seller $i'$ that can maximize $\pi(\phi_{i'}, \phi_{-i'})$, by definition, $\pi(\phi_{i'}^*, \phi_{-i'}) > \pi(\phi_{i'}^{**}, \phi_{-i'})$.
	Then, by maximizing the social welfare, we have

	\begin{align*}
		\sum_{i = 1}^{i'}\pi(\phi_i, \phi_{-i}) + \pi(\phi_{i'}^*, \phi_{-i'}) + \sum_{i = i'}^M\pi(\phi_i, \phi_{-i}) &>\\ \sum_{i = 1}^{i'}\pi(\phi_i, \phi_{-i}) + \pi(\phi_{i'}^{**}, \phi_{-i'}) + \sum_{i = i'}^M\pi(\phi_i, \phi_{-i}),
	\end{align*}
	which means there exists another $O' = \sum_{i = 1}^{i'} \phi_i^* + \phi_{i'} + \sum_{i = i'}^{M} \phi_i^* < O^{**}$.
	This contradicts with our assumption that $O^{**}$ is the social optimum.
\end{proof}

The corollary illustrates the problem of commons~\cite{gibbons1992game}, where the common resource is over-utilized because each seller considers only his or her own incentives, not the influence of his or her actions on the other sellers.
In general, cooperative behavior is more desirable for the general good, but competitive behavior may bring higher individual gain.

\section{Conclusion}
We formally define the bribing behaviors of sellers in the rating system and analyze its effect leveraging  game theory, where the sellers and their bribing behaviors are mapped into players and their corresponding strategy space, respectively.
In the static game, where the sellers and buyers are static and their payoffs are common knowledge, we find Nash equilibrium existing in the system when all sellers perform their best-bribing strategy.
In the equilibrium state, the system is bribery-proof and has the problem of commons~\cite{gibbons1992game} among sellers.
For our future work, we will design a simulation framework to evaluate our previous theoretical analysis.

\\

\noindent \textbf{Acknowledgements.} 
This work was supported in part by the National Natural Science Foundation of China under Grant No. 61906174 and 62172085, in part by the China Postdoctoral Science Foundation under Grant No. 2020M672275, and in part by JSPS KAKENHI Grant Number JP19H04170. We would like to thank Gautham Prakash for his sharing of Google Play apps dataset.

\bibliographystyle{splncs04}
\bibliography{game}

\begin{thebibliography}{10}
\providecommand{\url}[1]{\texttt{#1}}
\providecommand{\urlprefix}{URL }
\providecommand{\doi}[1]{https://doi.org/#1}

\bibitem{cao2016complete}
Cao, X., Huang, W., Yu, Y.: A complete \& comprehensive movie review dataset
  (ccmr). In: Proceedings of the 39th International ACM SIGIR conference on
  Research and Development in Information Retrieval. pp. 661--664. ACM (2016)

\bibitem{faliszewski2009hard}
Faliszewski, P., Hemaspaandra, E., Hemaspaandra, L.A.: How hard is bribery in
  elections? Journal of Artificial Intelligence Research  \textbf{35},
  485--532 (2009)

\bibitem{faliszewski2009llull}
Faliszewski, P., Hemaspaandra, E., Hemaspaandra, L.A., Rothe, J.: Llull and
  copeland voting computationally resist bribery and constructive control.
  Journal of Artificial Intelligence Research  \textbf{35},  275--341 (2009)

\bibitem{gibbons1992game}
Gibbons, R.: Game theory for applied economists. Princeton University Press
  (1992)

\bibitem{grandi2020personalised}
Grandi, U., Stewart, J., Turrini, P.: Personalised rating. Autonomous Agents
  and Multi-Agent Systems  \textbf{34}(2),  1--38 (2020)

\bibitem{Grandi2016}
Grandi, U., Turrini, P.: A network-based rating system and its resistance to
  bribery. In: Proceedings of the Twenty-Fifth International Joint Conference
  on Artificial Intelligence (IJCAI). pp. 301--307. AAAI Press (2016)

\bibitem{helpman2001lobbying}
Helpman, E., Persson, T.: Lobbying and legislative bargaining. Advances in
  economic analysis \& policy  \textbf{1}(1) (2001)

\bibitem{jiang2013evolutionary}
Jiang, S., Zhang, J., Ong, Y.S.: An evolutionary model for constructing robust
  trust networks. In: Proceedings of the 2013 international conference on
  Autonomous agents and multi-agent systems. pp. 813--820. International
  Foundation for Autonomous Agents and Multiagent Systems (2013)

\bibitem{josang2007survey}
J{\o}sang, A., Ismail, R., Boyd, C.: A survey of trust and reputation systems
  for online service provision. Decision support systems  \textbf{43}(2),
  618--644 (2007)

\bibitem{lianju2011game}
Lianju, S., Luyan, P.: Game theory analysis of the bribery behavior.
  International Journal of Business and Social Science  \textbf{2}(8) (2011)

\bibitem{liu2020fraudt}
Liu, C., Zhong, Q., Ao, X., Sun, L., Lin, W., Feng, J., He, Q., Tang, J.: Fraud
  transactions detection via behavior tree with local intention calibration.
  In: Proceedings of the 26th ACM SIGKDD International Conference on Knowledge
  Discovery and Data Mining. pp. 3035--3043 (2020)

\bibitem{liu20213r}
Liu, Y., Zhou, X., Yu, H.: 3r model: A post-purchase context-aware reputation
  model to mitigate unfair ratings in e-commerce. Knowledge-Based Systems
  \textbf{231},  107441 (2021)

\bibitem{manzoor2017attackdive}
Manzoor, S., Luna, J., Suri, N.: Attackdive: diving deep into the cloud
  ecosystem to explore attack surfaces. In: 2017 IEEE International Conference
  on Services Computing (SCC). pp. 499--502. IEEE (2017)

\bibitem{nash1950equilibrium}
Nash, J.F.: Equilibrium points in n-person games. Proceedings of the National
  Academy of Sciences of the United States of America  \textbf{36}(1), ~48
  (1950)

\bibitem{ouffoue2017attack}
Ouffou{\'e}, G.L., Za{\"\i}di, F., Cavalli, A.R., Lallali, M.: An
  attack-tolerant framework for web services. In: 2017 IEEE International
  Conference on Services Computing (SCC). pp. 503--506. IEEE (2017)

\bibitem{Parkes2017Thwarting}
Parkes, D.C., Tylkin, P., Xia, L.: Thwarting vote buying through decoy ballots.
  In: Proceedings of the 16th Conference on Autonomous Agents and MultiAgent
  Systems (AAMAS). pp. 1679--1681 (2017)

\bibitem{ramos2020negative}
Ramos, G., Boratto, L., Caleiro, C.: On the negative impact of social influence
  in recommender systems: A study of bribery in collaborative hybrid
  algorithms. Information Processing \& Management  \textbf{57}(2),  102058
  (2020)

\bibitem{resnick2006value}
Resnick, P., Zeckhauser, R., Swanson, J., Lockwood, K.: The value of reputation
  on ebay: A controlled experiment. Experimental economics  \textbf{9}(2),
  79--101 (2006)

\bibitem{sampath2018corporate}
Sampath, V.S., Gardberg, N.A., Rahman, N.: Corporate reputation’s invisible
  hand: Bribery, rational choice, and market penalties. Journal of Business
  Ethics  \textbf{151}(3),  743--760 (2018)

\bibitem{saude2021robust}
Sa{\'u}de, J., Ramos, G., Boratto, L., Caleiro, C.: A robust reputation-based
  group ranking system and its resistance to bribery. ACM Transactions on
  Knowledge Discovery from Data (TKDD)  \textbf{16}(2),  1--35 (2021)

\bibitem{wang2015quantifying}
Wang, D., Muller, T., Zhang, J., Liu, Y.: Quantifying robustness of trust
  systems against collusive unfair rating attacks using information theory. In:
  IJCAI. pp. 111--117 (2015)

\bibitem{wang2016harmful}
Wang, D., Muller, T., Zhang, J., Liu, Y.: Is it harmful when advisors only
  pretend to be honest? In: AAAI. pp. 2551--2557 (2016)

\bibitem{wang2012serf}
Wang, G., Wilson, C., Zhao, X., Zhu, Y., Mohanlal, M., Zheng, H., Zhao, B.Y.:
  Serf and turf: crowdturfing for fun and profit. In: Proceedings of the 21st
  international conference on World Wide Web. pp. 679--688. ACM (2012)

\bibitem{xu2017online}
Xu, C., Zhang, J., Sun, Z.: Online reputation fraud campaign detection in user
  ratings. In: Proceedings of the Twenty-Sixth International Joint Conference
  on Artificial Intelligence, (IJCAI). pp. 3873--3879 (2017)

\bibitem{ye2009impact}
Ye, Q., Law, R., Gu, B.: The impact of online user reviews on hotel room sales.
  International Journal of Hospitality Management  \textbf{28}(1),  180--182
  (2009)

\bibitem{zhou2016evaluating}
Zhou, X., Lin, D., Ishida, T.: Evaluating reputation of web services under
  rating scarcity. In: 2016 IEEE International Conference on Services Computing
  (SCC). pp. 211--218. IEEE (2016)

\bibitem{zhou2019arm}
Zhou, X., Murakami, Y., Ishida, T., Liu, X., Huang, G.: Arm: Toward adaptive
  and robust model for reputation aggregation. IEEE Transactions on Automation
  Science and Engineering  (2019)

\bibitem{zhu2014discovery}
Zhu, H., Xiong, H., Ge, Y., Chen, E.: Discovery of ranking fraud for mobile
  apps. IEEE Transactions on knowledge and data engineering  \textbf{27}(1),
  74--87 (2014)

\end{thebibliography}

\end{document}